\documentclass[11pt]{article}
\usepackage{amsmath,amsthm,amsfonts}
\DeclareMathOperator{\Hom}{Hom}
\DeclareMathOperator{\Com}{Com}

\DeclareMathOperator{\II}{I}
\DeclareMathOperator{\1}{id}

\newcommand{\NN}{\mathbb{N}}
\newcommand{\RR}{\mathbb{R}}
\newcommand{\CC}{\mathbb{C}}
\newcommand{\EEnd}{\mathcal End}
\newcommand{\EE}{\mathcal E}
\newcommand{\bul}{\bullet}

\renewcommand{\=}{:=}
\renewcommand{\t}{\otimes}

\newtheorem{thm}{Theorem}[section]
 \newtheorem{cor}[thm]{Corollary}
\theoremstyle{definition}
 \newtheorem{defn}[thm]{Definition}
\theoremstyle{definition}
 \newtheorem{exam}[thm]{Example}
\theoremstyle{definition}
 
\begin{document}
\title{\LARGE\bf Note on operadic harmonic oscillator}
\date{}
\author{\Large Eugen Paal and J\"{u}ri Virkepu\\ \\
Department of Mathematics, Tallinn University of Technology\\
Ehitajate tee 5, 19086 Tallinn, Estonia\\ \\
E-mails: eugen.paal@ttu.ee and jvirkepu@staff.ttu.ee}
\maketitle
\thispagestyle{empty}
\begin{abstract}
It is explained how the time evolution of the operadic variables may be introduced.
As an example, an operadic Lax representation for the  harmonic oscillator is constructed.
\end{abstract}

\section{Introduction and outline of the paper}

It is well known that quantum mechanical observables are linear \emph{operators}, i.e the linear maps $V\to V$ of a vector space $V$ and their time evolution is given by the Heisenberg equation. As a variation of this one can pose the following question \cite{Paal07}: how to describe the time evolution of the  linear algebraic operations (multiplications) $V^{\t n}\to V$. The algebraic operations (multiplications) can be seen as an example of the \emph{operadic} variables \cite{Ger,GGS92,KP,KPS}. 

When an operadic system depends on time one can speak about \emph{operadic dynamics} \cite{Paal07}. 
The latter may be introduced by simple and natural analogy with the Hamiltonian dynamics.
In particular, the time evolution of operadic variables may be given by operadic Lax equation.
As an example, in the present paper, an operadic Lax representation for the harmonic oscillator is constructed.

\section{Operad}

Let $K$ be a unital associative commutative ring, and let $C^n$ ($n\in\NN$) be unital $K$-modules. For  $f\in C^n$, we refer to $n$ as the \emph{degree} of $f$ and often write (when it does not cause confusion) $f$ instead of $\deg f$. For example, $(-1)^f\=(-1)^n$, $C^f\=C^n$ and $\circ_f\=\circ_n$. Also, it is convenient to use the \emph{reduced} degree $|f|\=n-1$. Throughout this paper, we assume that $\t\=\t_K$.

\begin{defn}[operad (e.g \cite{Ger,GGS92})]
A linear (non-symmetric) \emph{operad} with coefficients in $K$ is a sequence $C\=\{C^n\}_{n\in\NN}$ of unital
$K$-modules (an $\NN$-graded $K$-module), such that the following
conditions are held to be true.
\begin{enumerate}
\item[(1)]
For $0\leq i\leq m-1$ there exist \emph{partial compositions}
\[
  \circ_i\in\Hom(C^m\t C^n,C^{m+n-1}),\qquad |\circ_i|=0
\]
\item[(2)]
For all $h\t f\t g\in C^h\t C^f\t C^g$,
the \emph{composition (associativity) relations} hold,
\[
(h\circ_i f)\circ_j g=
\begin{cases}
    (-1)^{|f||g|} (h\circ_j g)\circ_{i+|g|}f
                       &\text{if $0\leq j\leq i-1$},\\
    h\circ_i(f\circ_{j-i}g)  &\text{if $i\leq j\leq i+|f|$},\\
    (-1)^{|f||g|}(h\circ_{j-|f|}g)\circ_i f
                       &\text{if $i+f\leq j\leq|h|+|f|$}.
\end{cases}
\]
\item[(3)]
Unit $\II\in C^1$ exists such that
\[
\II\circ_0 f=f=f\circ_i \II,\qquad 0\leq i\leq |f|
\]
\end{enumerate}
\end{defn}

In the second item, the \emph{first} and \emph{third} parts of the
defining relations turn out to be equivalent.

\begin{exam}[endomorphism operad \cite{Ger}]
\label{HG} Let $V$ be a unital $K$-module and
$\EE_V^n\={\EEnd}_V^n\=\Hom(V^{\t n},V)$. Define the partial compositions
for $f\t g\in\EE_V^f\t\EE_V^g$ as
\[
f\circ_i g\=(-1)^{i|g|}f\circ(\1_V^{\t i}\t g\t\1_V^{\t(|f|-i)}),
         \qquad 0\leq i\leq |f|
\]
Then $\EE_V\=\{\EE_V^n\}_{n\in\NN}$ is an operad (with the unit $\1_V\in\EE_V^1$) called the 
\emph{endomorphism operad} of $V$.

Therefore, algebraic operations can be seen as elements of an endomorphism operad.
\end{exam}

Just as elements of a vector space are called \emph{vectors},  it is natural to call elements of an abstract operad \emph{operations}. The endomorphism operads can be seen as the most suitable objects for modelling operadic systems.

\section{Gerstenhaber brackets and operadic Lax pair} 

\begin{defn}[total composition \cite{Ger,GGS92}]
The \emph{total composition} $\bul\:C^f\t C^g\to C^{f+|g|}$ is defined by
\[
f\bul g\=\sum_{i=0}^{|f|}f\circ_i g\in C^{f+|g|},
\qquad |\bul|=0
\]
The pair $\Com C\=\{C,\bul\}$ is called the \emph{composition algebra} of $C$.
\end{defn}

\begin{defn}[Gerstenhaber brackets \cite{Ger,GGS92}]
The  \emph{Gerstenhaber brackets} $[\cdot,\cdot]$ are defined in $\Com C$ as a graded commutator by
\[
[f,g]\=f\bul g-(-1)^{|f||g|}g\bul f=-(-1)^{|f||g|}[g,f],\qquad|[\cdot,\cdot]|=0 
\]
\end{defn}

The \emph{commutator algebra} of $\Com C$ is denoted as $\Com^{-}\!C\=\{C,[\cdot,\cdot]\}$. 
One can prove that $\Com^-\!C$ is a \emph{graded Lie algebra}. The Jacobi
identity reads
\[
(-1)^{|f||h|}[[f,g],h]+(-1)^{|g||f|}[[g,h],f]+(-1)^{|h||g|}[[h,f],g]=0 
\]

Assume that $K\=\RR$ or $K\=\CC$ and operations are differentiable.
The operadic dynamics may be introduced by the

\begin{defn}[operadic Lax pair \cite{Paal07}]
Allow a classical dynamical system to be described by the evolution equations 
\[
\dfrac{dx_i}{dt}=f_i(x_1,\dots,x_n),\quad i=1,\dots,n
\]
An \emph{operadic Lax pair} is a pair $(L,M)$ of homogeneous operations $L,M\in C$, 
such that the above system of evolution equations may be written as the
\emph{operadic Lax equation}
\[
\dfrac{dL}{dt}=[M,L]\=M\bul L-(-1)^{|M||L|}L\bul M
\]
Evidently,  the degree constraints $|M|=|L|=0$ give rise to ordinary Lax pair \cite{Lax68,BBT03}.
\end{defn}

\section{Operadic harmonic oscillator}

Surprisingly, examples are at hand. One can use the Lax pairs to extend these to operadic area via the operadic Lax equation. 

Consider the Lax pair for the harmonic oscillator:
\[
L=\begin{pmatrix}
p&\omega q\\
\omega q &-p
\end{pmatrix},
\qquad
M=\frac{\omega}{2}
\begin{pmatrix}
0&-1\\
1&0
\end{pmatrix}
\]
Since the Hamiltonian is
\[
H(q,p)=\frac{1}{2}(p^2+\omega^2q^2),
\]
it is easy to check that the Lax equation 
\[
\dot{L}=[M,L]\= ML - LM.
\]
may be written as the Hamiltonian system
\[
\dfrac{dq}{dt}=\dfrac{\partial H}{\partial p}=p,
\quad
\dfrac{dp}{dt}=-\dfrac{\partial H}{\partial q}=-\omega^2q.
\]
If $\mu$ is a homogeneous operadic variable one can use these Hamilton's equations to calculate
\[
\dfrac{d\mu}{dt}
=\dfrac{\partial\mu}{\partial q}\dfrac{dq}{dt}+\dfrac{\partial\mu}{\partial p}\dfrac{dp}{dt}
=p\dfrac{\partial\mu}{\partial q}-\omega^2q\dfrac{\partial\mu}{\partial p}
=[M,\mu]
\]
from which it follows the following linear partial differential equation for the operadic variable $\mu(q,p)$:
\[
p\dfrac{\partial\mu}{\partial q}-\omega^2q\dfrac{\partial\mu}{\partial p}=M\bul\mu- \mu\bul M.
\]
By integrating one gains sequences of operations  called the \emph{operadic (representations of) harmonic oscillator}.

\section{Example}

Let $A\=\{V,\mu\}$ be a  binary algebra with operation $xy\=\mu(x\t y)$. 
We require that $\mu=\mu(q,p)$ so that $(\mu,M)$ is an operadic Lax pair, i.e the operadic Lax equation
\[
\dot{\mu}=[M,\mu]\= M\bul\mu-\mu\bul M,\qquad |\mu|=1,\quad |M|=0
\]
is equivalent to the Hamiltonian system of the harmonic oscillator.

Let $x,y\in V$. By assuming that $|M|=0$ and $|\mu|=1$, one has
\begin{align*}
M\bul\mu
&=\sum_{i=0}^0(-1)^{i|\mu|}M\circ_i\mu
=M\circ_0\mu\\
&=M\circ\mu\\
\mu\bul M
&=\sum_{i=0}^1(-1)^{i|M|}\mu\circ_i M
=\mu\circ_0 M+\mu\circ_1 M\\
&=\mu\circ(M\t\1_V)+\mu\circ(\1_V\t M)
\end{align*}
Therefore, one has
\[
\dfrac{d}{dt}(xy)=M(xy)-(Mx)y-x(My)
\]
Let $\dim V=n$.
In a basis $\{e_1,\ldots,e_n\}$ of $V$,  the structure constants $\mu_{jk}^i$ of $A$ are defined by
\[
\mu(e_j\t e_k)\= \mu_{jk}^i e_i,\qquad j,k=1,\ldots,n
\]
In particular,
\[
\dfrac{d}{dt}(e_je_k)=M(e_je_k)-(Me_j)e_k-e_j(Me_k)
\]
By denoting $Me_i\= M_i^se_s$, it follows that
\[
\dot{\mu}_{jk}^i=\mu_{jk}^sM_s^i-M_j^s\mu_{sk}^i-M_k^s\mu_{js}^i,\qquad i,j,k=1,\ldots, n
\]
From these equations it follows the
\begin{thm}
Let $\dim V=2$.
By identifying 
$
M\=(M_j^i)\=
\frac{\omega}{2}
\left(
\begin{smallmatrix}
0&-1\\
1&0
\end{smallmatrix}
\right)
$, 
the $2$-dimensional binary operadic Lax equations for the harmonic oscillator read
\[
\begin{cases}
\dot{\mu}_{11}^{1}=-\frac{\omega}{2}\left(\mu_{11}^{2}+\mu_{21}^{1}+\mu_{12}^{1}\right),\qquad
\dot{\mu}_{11}^{2}=\frac{\omega}{2}\left(\mu_{11}^{1}-\mu_{21}^{2}-\mu_{12}^{2}\right)\\
\dot{\mu}_{12}^{1}=-\frac{\omega}{2}\left(\mu_{12}^{2}+\mu_{22}^{1}-\mu_{11}^{1}\right),\qquad
\dot{\mu}_{12}^{2}=\frac{\omega}{2}\left(\mu_{12}^{1}-\mu_{22}^{2}+\mu_{11}^{2}\right)\\
\dot{\mu}_{21}^{1}=-\frac{\omega}{2}\left(\mu_{21}^{2}-\mu_{11}^{1}+\mu_{22}^{1}\right),\qquad
\dot{\mu}_{21}^{2}=\frac{\omega}{2}\left(\mu_{21}^{1}+\mu_{11}^{2}-\mu_{22}^{2}\right)\\
\dot{\mu}_{22}^{1}=-\frac{\omega}{2}\left(\mu_{22}^{2}-\mu_{12}^{1}-\mu_{21}^{1}\right),\qquad
\dot{\mu}_{22}^{2}=\frac{\omega}{2}\left(\mu_{22}^{1}+\mu_{12}^{2}+\mu_{21}^{2}\right)\\
\end{cases}
\]
\end{thm}
In what follows, consider only  anti-commutative algebras. Then one has the
\begin{cor}
Let $A$ be a $2$-dimensional anti-commutative algebra, i.e
\[
\mu_{11}^{1}=\mu_{22}^{1}=\mu_{11}^{2}=\mu_{22}^{2}=0,\quad
\mu_{12}^{1}=-\mu_{21}^{1},\quad
\mu_{12}^{2}=-\mu_{21}^{2}
\]
Then the operadic Lax equations read
\[
\begin{cases}
\dot{\mu}_{12}^1=-\frac{\omega}{2}{\mu}_{12}^2\\
\dot{\mu}_{12}^2=\hphantom{-}\frac{\omega}{2}{\mu}_{12}^1\\
\end{cases}
\]
\end{cor}
Thus,  one has to specify $\mu_{12}^1$ and $\mu_{12}^2$ as functions of the canonical variables $q$ and $p$.
Define
\[
\begin{cases}
A_+\= \sqrt{\sqrt{2H}+p}\\
A_-\=  \sqrt{\sqrt{2H}-p}\\
\end{cases}
\]
and
\[
\begin{cases}
B_+\= A_++A_-=\sqrt{\sqrt{2H}+p}+\sqrt{\sqrt{2H}-p}\\ 
B_-\= A_+-A_-=\sqrt{\sqrt{2H}+p}-\sqrt{\sqrt{2H}-p}\\
\end{cases}
\]
Then one has 
\begin{thm}
The formulae
\[
M=\frac{\omega}{2}
\begin{pmatrix}
0&-1\\
1&0
\end{pmatrix},
\qquad
\begin{cases}
\mu_{11}^{1}=\mu_{22}^{1}=\mu_{11}^{2}=\mu_{22}^{2}=0\\
{\mu}_{12}^1=-{\mu}_{21}^1=B_-\\
{\mu}_{12}^2=-{\mu}_{21}^2=B_+
\end{cases}
\]
represent a $2$-dimensional binary operadic Lax pair of the harmonic oscillator. 
The algebra given by the above structure functions $\mu_{jk}^i$ is a $2$-dimensional Lie algebra.
\end{thm}

\begin{proof}
The operadic Lax equations read
\[
\begin{cases}
\dot{B}_-=-\frac{\omega}{2}B_+\\
\dot{B}_+=\hphantom{-}\frac{\omega}{2}B_-\\
\end{cases}
\]
That is
\[
\begin{cases}
\left[\frac{1}{2A_+}\left(\frac{p}{\sqrt{2H}}+1\right)-\frac{1}{2A_-}\left(\frac{p}{\sqrt{2H}}-1\right)\right]\dot{p}
+\left[\left(\frac{1}{2A_+}-\frac{1}{2A_-}\right)\frac{q\omega^{2}}{\sqrt{2H}}\right]\dot{q}
=-\frac{\omega}{2}B_+\vspace{1mm}\\
\left[\frac{1}{2A_+}\left(\frac{p}{\sqrt{2H}}+1\right)+\frac{1}{2A_-}\left(\frac{p}{\sqrt{2H}}-1\right)\right]\dot{p}+\left[\left(\frac{1}{2A_+}+\frac{1}{2A_-}\right)\frac{q\omega^{2}}{\sqrt{2H}}\right]\dot{q}
=\hphantom{-}\frac{\omega}{2}B_-\\
\end{cases}
\]
Multiplying both equations by $2A_+A_-$ one gets
\[
\begin{cases}
\left[A_-\left(\frac{p}{\sqrt{2H}}+1\right)-A_+\left(\frac{p}{\sqrt{2H}}-1\right)\right]\dot{p}
-\frac{q\omega^{2}B_-}{\sqrt{2H}}\dot{q}
=-\omega B_+A_+A_-\vspace{1mm}\\
\left[A_-\left(\frac{p}{\sqrt{2H}}+1\right)+A_+\left(\frac{p}{\sqrt{2H}}-1\right)\right]\dot{p}
+\frac{q\omega^{2}B_+}{\sqrt{2H}}\dot{q}
=\hphantom{-}\omega B_-A_+A_-\\
\end{cases}
\]
Now use the Cramer formulae. 
By using the relations
\[
B_+^2-B_-^2=4A_+A_-,
\qquad 
(A_+A_-)^2=q^2\omega^2
\]
first calculate the determinants
\begin{align*}
\Delta
&=\begin{vmatrix}
         A_-\left(\frac{p}{\sqrt{2H}}+1\right)-A_+\left(\frac{p}{\sqrt{2H}}-1\right) &
-\frac{q\omega^{2}B_-}{\sqrt{2H}}\vspace{1mm}\\
         A_-\left(\frac{p}{\sqrt{2H}}+1\right)+A_+\left(\frac{p}{\sqrt{2H}}-1\right) &
 \hphantom{-}\frac{q\omega^{2}B_+}{\sqrt{2H}} \\
\end{vmatrix}
=\frac{4q^{2}\omega^{3}}{\sqrt{2H}}\\
\Delta_{\dot{p}}
&=\begin{vmatrix}
-\omega B_+A_+A_-& -\frac{q\omega^2B_-}{\sqrt{2H}}\vspace{1mm} \\
\hphantom{-}\omega B_-A_+A_-& \hphantom{-}\frac{q\omega^2B_+}{\sqrt{2H}} \\
\end{vmatrix}=-\frac{4q^{3}\omega^{5}}{\sqrt{2H}}\\
\Delta_{\dot{q}}
&=\begin{vmatrix}
A_-\left(\frac{p}{\sqrt{2H}}+1\right)-A_+\left(\frac{p}{\sqrt{2H}}-1\right) & -\omega B_+A_+A_- \vspace{1mm}\\
A_-\left(\frac{p}{\sqrt{2H}}+1\right)+A_+\left(\frac{p}{\sqrt{2H}}-1\right) & \hphantom{-}\omega B_-A_+A_-\\
\end{vmatrix}
=\frac{4pq^{2}\omega^{3}}{\sqrt{2H}}
\end{align*}
Thus one obtains the Hamiltonian system of the harmonic oscillator,
\[
\dot{q}=\frac{\Delta_{\dot{q}}}{\Delta}=p,\qquad
\dot{p}=\frac{\Delta_{\dot{p}}}{\Delta}=-q\omega^{2}
\]
and the latter is equivalent to the above operadic Lax system of the harmonic oscillator.

The Jacobi identity for $\mu_{jk}^i$ can be checked by direct calculation.
\end{proof}

\section*{Acknowledgement}
Research was in part supported by the Estonian Science Foundation, Grant 6912.

\end{document}